\documentclass[letter,scriptaddress,twocolumn, showkeys]{revtex4}
\usepackage{ amssymb }
	\usepackage{amsmath}
	\usepackage{makeidx}
	\usepackage{amsfonts}
	\usepackage[ansinew]{inputenc}
	\usepackage[usenames,dvipsnames]{pstricks}
	\usepackage{subfigure}
	\usepackage{epsfig}
	\usepackage{pst-grad} 
	\usepackage{pst-plot} 
		\usepackage{amsthm}
\usepackage{ amssymb }
 	\usepackage{makeidx}
	\usepackage{amsfonts}
 	\usepackage{listings}
 	\usepackage{mathtools}
	\usepackage{float}
	
	\usepackage[colorlinks,hyperindex]{hyperref}
	\hypersetup
	{
		colorlinks,%
		citecolor=blue,%
		linkcolor=blue,%
		urlcolor=blue,%
	}
	\usepackage{tikz}

\DeclarePairedDelimiter{\ceil}{\lceil}{\rceil}

\theoremstyle{definition}
	\newtheorem{definition}{Definition}
	\newtheorem{example}[definition]{Example}
	
	\theoremstyle{plain}
	\newtheorem{theorem}[definition]{Theorem}
	
	\newtheorem{prop}[definition]{Proposition}
  	\newtheorem{lemma}[definition]{Lemma}
 	
 	\newtheorem{remark}[definition]{Remark}

\makeindex

\begin{document}

\title{Minimal percolating sets for mutating infectious diseases}

\author{Yuyuan Luo$^{a}$ and   Laura P. Schaposnik$^{b,c}$}

  \affiliation {(a)  Central High School, Grand Rapids, MI 49546, USA. \\
  (b)  University of Illinois, Chicago, IL 60607, USA.\\
    (c)  Mathematical Sciences Research Institute, Berkeley, CA 94720, USA.}

\begin{abstract}

This paper is dedicated to the study of the interaction between dynamical systems and percolation models, with views towards the study of viral infections whose virus mutate with time. Recall that $r$-bootstrap percolation describes a deterministic process where vertices of a graph are infected once $r$ neighbors of it are infected. We generalize this by introducing {\it $F(t)$-bootstrap percolation}, a time-dependent process where  the number of neighbouring vertices which need to be infected for a disease to be transmitted is determined by a percolation function $F(t)$ at each time $t$. After studying some of the basic properties of the model, we consider smallest percolating sets and construct a polynomial-timed algorithm to find one smallest minimal percolating set on finite trees for certain $F(t)$-bootstrap percolation models. \\

\end{abstract}

\maketitle
 
\section{Introduction}
The   study infectious diseases though mathematical models dates back to 1766, where   Bernoulli developed a model to examine the mortality due to smallpox in England    \cite{modeling}. Moreover, the germ theory that describes the spreading of infectious diseases was first established in 1840 by  Henle and was further developed in the late 19th and early 20th centuries. This laid the groundwork for mathematical models as it explained the way that infectious diseases spread, which led to the rise of compartmental models. These models divide populations into compartments, where individuals in each compartment have the same characteristics; Ross first established one such model in 1911 in \cite{ross} to study malaria  and later on, basic compartmental models to study infectious diseases were established in a sequence of three papers by Kermack and McKendrick  \cite{kermack1927contribution} (see also \cite{epidemiology} and references therein). 

In these notes  we are interested in the interaction between dynamical systems and percolation models, with views towards the study of infections which mutate with time.  The use of stochastic models to study infectious diseases dates back to 1978 in work of  J.A.J. Metz   \cite{epidemiology}.
There are many ways to mathematically model infections, including statistical-based models such as regression models (e.g.~\cite{imai2015time}), cumulative sum charts (e.g.~\cite{chowell2018spatial}), hidden Markov models (e.g.~\cite{watkins2009disease}), and spatial models (e.g.~\cite{chowell2018spatial}), as well as mechanistic state-space models such as continuum models with differential equations (e.g.~\cite{greenhalgh2015disease}), stochastic models (e.g.~\cite{pipatsart2017stochastic}), complex network models (e.g.~\cite{ahmad2018analyzing}), and agent-based simulations (e.g.~\cite{hunter2019correction}  -- see also \cite{modeling} and references therein).

Difficulties when modeling infections include incorporating the dynamics of behavior in models, as it may be difficult to access the extent to which behaviors should be modeled explicitly, quantify changes in reporting behavior, as well as identifying the role of movement and travel   \cite{challenges}. When using data from multiple sources, difficulties may arise when determining how the evidence should be weighted and when handling dependence between datasets   \cite{challenges2}.

 In what follows we shall introduce a novel type of dynamical percolation which we call {\it $F(t)$-bootstrap percolation},   though a generalization of classical bootstrap percolation. This approach allows one to model mutating infections, and thus we dedicate this paper to the study some of its main features. After recalling classical $r$-bootstrap percolation in Section \ref{intro}, we introduce a percolating function $F(t)$ through which we introduce a dynamical aspect the percolating model, as described in Definition \ref{fperco}.
 \smallbreak
 
\noindent {\bf Definition.} Given a function $F(t): \mathbb{N}\rightarrow \mathbb{N}$, we define an {\em $F(t)$-bootstrap percolation model} on a graph $G$ with vertices $V$ and initially infected set $A_0$  as the process  which at time $t+1$ has infected set given by   
  \begin{eqnarray}A_{t+1} = A_{t} \cup \{v \in V  : |N(v) \cap A_t| \geq F(t)\}, \end{eqnarray} 
 where $N(v)$ denotes the set of neighbouring vertices to $v$, and we let    $A_\infty$ be the final set of infected vertices once the percolation process has finished. 
\smallbreak

In Section \ref{time} we study some basic properties of this model,   describe certain (recurrent) functions which ensure the model percolates, and study the critical probability $p_c$. Since our motivation comes partially from the study of effective vaccination programs which would allow to contain an epidemic, we are interested both in the percolating time of the model, as well as in minimal percolating sets. We study the former in Section \ref{time2}, where by considering equivalent functions to $F(t)$, we obtained bounds on the percolating time in Proposition \ref{propo8}. 
\smallbreak

Finally, in Section \ref{minimal} and Section \ref{minimal2} we introduce and study smallest minimal percolating sets for $F(t)$-bootstrap percolation on (non-regular) trees. This leads to one of our main results in Theorem \ref{teo1}, where we describe an algorithm for finding the smallest minimal percolating sets.  Lastly, we conclude the paper with a comparison in Section \ref{final} of our model and algorithm to the model and algorithm considered in \cite{percset} for clasical bootstrap percolation, and analyse the effect of taking different functions within our dynamical percolation. 

\newpage

\section{Background: bootstrap percolation and SIR models}\label{intro}

Bootstrap percolation was introduced in 1979 in the context of solid state physics in order to analyze diluted magnetic systems in which strong competition exists between exchange and crystal-field interactions   \cite{density}.
It has seen applications in the studies of fluid flow in porous areas, the orientational ordering process of magnetic alloys, as well as the failure of units in a structured collection of computer memory   \cite{applications}.

Bootstrap percolation has long been studied mathematically on finite and infinite rooted trees including Galton-Watson trees (e.g. see  \cite{MR3164766}).
It better simulates the effects of individual behavior and the spatial aspects of epidemic spreading, and better accounts for the effects of mixing patterns of individuals. Hence,  communicative diseases in which these factors have significant effects are better understood when analyzed with cellular automata models such as bootstrap percolation   \cite{automata}, which is defined as follows. 

\begin{definition}[Bootstrap percolation] For $n\in \mathbb{Z}^+$, we define an {\em $n$-bootstrap percolation model} on a graph $G$ with vertices $V$ and initially infected set $A_0$  as the process in which at time $t+1$ has infected set given by   
  \begin{eqnarray}A_{t+1} = A_{t} \cup \{v \in V  : |N(v) \cap A_t| \geq n\}. \end{eqnarray} 
 Here, as before, we denoted by $N(v)$ the set of neighbouring vertices to $v$.   
\end{definition}

In contrast, a {\it SIR Model} relates at each time $t$  the number of susceptible individuals $S(t)$ with  the number of infected individuals $I(t)$  and  the number of recovered individuals $R(t)$, by a system of differential equations -- an example of a SIR model   used  to simulate the spread of the dengue fever disease appears in  \cite{dengue}. 
The SIR models are very useful for simulating infectious diseases; however, compared to bootstrap percolation, SIR models do not account for individual behaviors and characteristics. In these models, a fixed parameter $\beta$ denotes the average number of transmissions from an infected node in a time period.
 
In what follows we shall present a dynamical generalization of the above model, for which it will be useful to have an example to establish the comparisons. 
 \begin{figure}[h!]
\includegraphics[scale=.26]{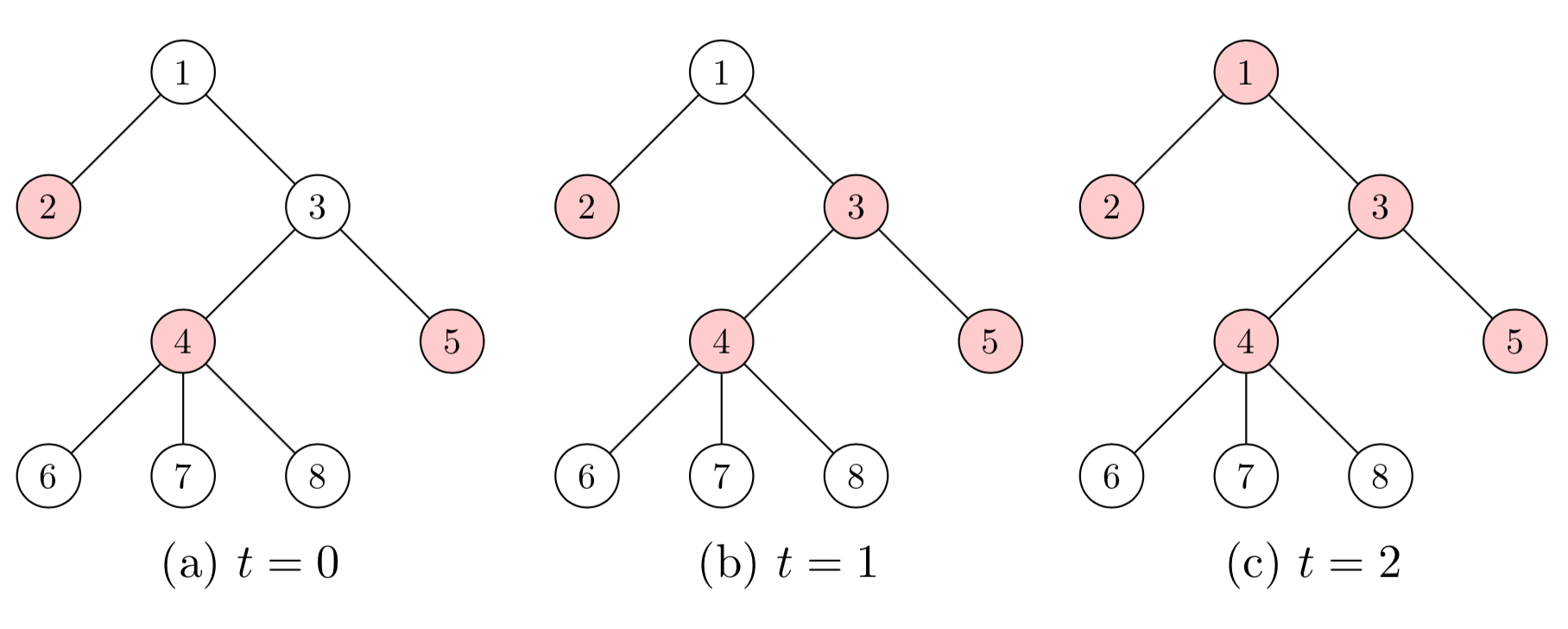}
\caption{Depiction of $2$-bootstrap percolation, where shaded vertices indicated infected nodes. }\label{first}
\end{figure}

Consider the  (irregular) tree  with three infected nodes at time $t=0$, given by $A_0=\{2,4,5\}$ as shown in Figure \ref{first}. 
Then,  through  $2$-bootstrap percolation at time $t=1$, node $3$ becomes infected because its neighbors $4$ and $5$ are infected at time $t=0$. At time $t=2$, node $1$ becomes  infected since its neighbors $2$ and $3$ are infected at time $t=1$. Finally, note that nodes $6,7,8$ cannot become infected because they each have only $1$ neighbor, yet two or more infected neighbors are required to become infected.

\section{Time-dependent Percolation }\label{time}

The motivation of time-dependent percolation models appears since   the rate of spread of diseases may change over time. In the SIR models mentioned before, since $\beta$ is   the average number of transmissions from an infected node in a time period, $1/\beta$ is the time it takes to infect a node. If we ``divide the work" among several neighbors, then $1/\beta$ is also the number of infected neighbors needed to infect the current node.  
Consider now an infection which would evolve with time. This is, instead of taking the same number of neighbours in $r$-bootstrap percolation, consider a percolation model where the number of neighbours required to be infected for the disease to propagate changes with time, following the behaviour of a  function $F(t)$ which can be set in terms of a one-parameter family of parameters $\beta$  to be $F(t) := \ceil[bigg]{\frac{1}{\beta(t)}}$. 
 We shall say a function is a  {\it percolation function} if it is a function $F: I \rightarrow \mathbb{Z}^+$ where $I$ is an initial segment of $\mathbb{N}$ that we use in a time-dependent percolation process, and which specifies the number of neighbors required to percolate to a node at time $t$. 
\begin{definition}[$F(t)$-Bootstrap percolation] \label{fperco}Given a function $F(t): \mathbb{N}\rightarrow \mathbb{N}$, we define an {\em $F(t)$-bootstrap percolation model} on a graph $G$ with vertices $V$ and initially infected set $A_0$  as the process in which at time $t+1$ has infected set given by   
  \begin{eqnarray}A_{t+1} = A_{t} \cup \{v \in V  : |N(v) \cap A_t| \geq F(t)\}. \end{eqnarray} 
 Here, as before, we denoted by $N(v)$ the set of neighbouring vertices to $v$, and we let    $A_\infty$ be the final set of infected vertices once the percolation process has finished. 
\end{definition}

\begin{remark}One should note that   $r$-bootstrap percolation can be recovered from $F(t)$-bootstrap percolation by setting  the percolation function to be the constant $F(t) = r$.
\end{remark}
 
It should be noted that, unless otherwise stated,  the initial set $A_0$ is chosen in the same way as in $r$-bootstrap percolation: by randomly selecting a set of initially infected vertices with probability $p$, for some fixed value of $p$ which is called the {\it probability of infection}.  
If there are multiple percolation functions and initially infected sets in question, we may use the notation $A^{F }_{t}$ to denote the set of infected nodes at time $t$ percolating under the function $F(t)$ with $A_0$ as the initially infected set. In particular, this would be the case when implementing the above dynamical model to a multi-type bootstrap percolation such as the one introduced in \cite{gossip}.  
In order to understand some basic properties of $F(t)$-bootstrap percolation, we shall first focus on a single update function $F(t)$, and consider   the critical probability   $p_c$   of infection for which the probability of percolation is $\frac{1}{2}$.

\begin{prop}\label{propo1}
If $F(t)$ equals its minimum for infinitely many times $t$, then the  critical probability  of infection $p_c$  for which the probability of percolation is 1/2, is given by the value of the critical probability in $m$-bootstrap percolation,  for $m:=\min_t F(t)$.\end{prop} 

\begin{proof}
 When considering classical bootstrap percolation, note that   the resulting set $A_\infty^r$ of $r$-bootstrap percolation is always contained by the resulting set $A_\infty^n$ of $n-$bootstrap percolation provided $n\leq r$. Hence, setting the value  $m:=\min_t F(t)$,    the resulting $A_\infty^F$ set of $F(t)$-bootstrap percolation will be contained in $A_\infty^m$. Moreover, since any vertex in $A_t^F$  for $t$ such that $F(t)=m$ remains in the set the next time for which $F(t)=m$, and since there are infinitely many times $t$ such that $F(t)=m$, we know that the final resulting set $A_\infty^m$ of $m$-bootstrap percolation is contained in the final resulting set $A_\infty^F$ of $F(t)$-bootstrap percolation. Then the resulting set of $m$-bootstrap percolation and $F(t)$-bootstrap percolation need to be identical, and hence the  critical probability for $F(t)$-bootstrap percolation is that of $m$-bootstrap percolation.
\end{proof}

As we shall see later, different choices of the one-parameter family $\beta(t)$ defining $F(t)$ will lead to very different dynamical models.  
A particular set up arises from  \cite{viral}, which provides data on the time-dependent rate of a specific virus spread, and through which  one has that an interesting family of parameters appears by setting \[\beta(t) = \left(b_0-b_f\right)\cdot\left(1-k\right)^t+b_f,\]   where $b_0$ is the initial rate of spread, $b_f$ is the final rate of spread, and $0<k<1$. Then at time $t$, the number of infected neighbors it takes to infect a node is \[F(t):=\ceil[Bigg]{\frac{1}{\left(b_0-b_f\right)\cdot\left(1-k\right)^t+b_f}}.\] 

In this case, since $\beta(t)$ tends to $b_f$, and $\frac{1}{\beta}$ tends to  $\frac{1}{b_f}$,  one cans see that  there will be infinitely many times $t$ such that $F(t) = \ceil[Bigg]{\frac{1}{b_f}}$. Hence, in this setting  from Proposition \ref{propo1},  the critical probability will be same as that of a $r$-bootstrap percolation where $r=\ceil[Bigg]{\frac{1}{b_f}}$.

\section{Percolation Time}\label{time2}

Informally, {\it percolation time} is the time it takes for the percolation process to terminate, with regards to a specific initially infected set of a graph. In terms of limits, recall that the final percolating set is defined as 
\begin{eqnarray}A_\infty:=\lim_{t\rightarrow \infty} A_t,\label{mas}\end{eqnarray}
and thus one may think of the percolation time as the smallest time $t$ for which $A_t=A_\infty$. By considering different initial probabilities of infection $p$ which determine the initially infected set $A_0$, and different percolation  functions $F(t)$ one can see that the percolation time of a model can vary drastically. To illustrate this, in Figure \ref{second} we have plotted the percentage of nodes infected with two different initial probabilities and four different percolation functions. The model was ran $10^3$ times for each combination on random graphs with $10^2$ nodes and $300$ edges.

 \begin{figure}[h!]
\includegraphics[scale=.4]{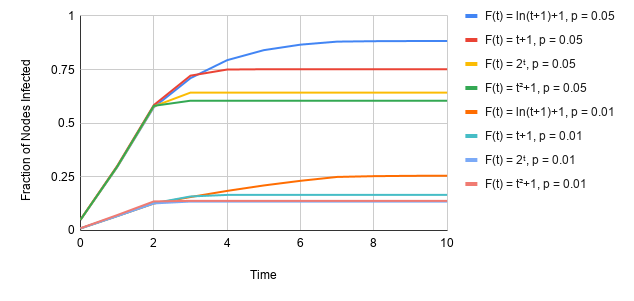}
\caption{ Percentage of nodes infected at time $t$ for $F(t)$-bootstrap percolation with initial probability $p$, on graphs with $100$ nodes and $300$ edges.}\label{second}
\end{figure}

 In the above settings of Figure \ref{second}, one can see that all the models stabilize by time $10$, implying that the percolation time is less than or equal to $10$.  Generally, understanding the percolation time is useful in determining when the disease spreading has stabilized. In what follows, we find a method to generate an upper bound on the percolation time given a specific graph and function. Formally, we define the 
  {\it percolation time} $t_*$ as the minimum 
  \[t_*:=\min_t \{~t~|~A_{t+1} = A_t~\}.\]   

Expanding on the notation of  \eqref{mas}, we shall denote by $A_\infty^\gamma$  the set of nodes infected by percolating the set $A_0$ on the graph with percolation function $\gamma(t)$, and we shall simply write $A_\infty$  when the percolation function $\gamma(t)$ is clear from context or irrelevant. Moreover, we shall say that 
two percolation functions $F_1: I_1 \rightarrow \mathbb{Z}^+$ and $F_2: I_2 \rightarrow \mathbb{Z}^+$   are   {\it equivalent} for the graph $G$ if for all initially infected sets $A_0$,  one has that \[A^{F_1}_\infty=A^{F_2}_\infty.\]
This equivalence relation can be understood through the lemma below, which uses an additional function $\gamma(t)$ to relate two percolation functions $F_0$ and $F_0'$ if $F_0'$ can be intuitively ``generated'' by removing some values of $F_0$. This removal procedure is further specified in this lemma.

Given two subsets $I_1$ and  $I_2$  of $\mathbb{N}$, we say a function   $\gamma: I_1 \rightarrow I_2 \cup \{-1\}$ is a {\it nice function} if it is surjective  and
\begin{itemize}
\item it is injective on $\gamma^{-1}(I_2)$;
\item it is increasing on $\gamma^{-1}(I_2)$;
\item it satisfies $\gamma(a) \leq a$ or $\gamma(a)=-1$.

\end{itemize}

\begin{lemma}
Given $I_1,I_2\subset \mathbb{N}$, let $F(t)$ be any percolation function with domain $I_1$, and define the percolation function $F'(t)$ with domain $I_2$  as $F'(t) := F(\gamma^{-1}(t))$ for $\gamma(t)$ a nice function. Then, for any fixed initially infected set $A_0$ and $t \in I_2$, one has that  \begin{eqnarray}A^{F'}_{t} \subseteq A^{F}_{\gamma^{-1}(t)}.\label{mas11}\end{eqnarray}
\end{lemma}

\begin{proof}
We first show that $F'(t)$ is well-defined. Since the domain of $F'(t)$ is $I_2$, we have that  $t\in I_2$ and thus $\gamma^{-1}(t)$ is a valid expression. Moreover,  $\gamma^{-1}(t)$ exists because $\gamma$ is surjective, and it is unique since $I_2$ is an initial segment of $\mathbb{N}$ and hence $t \neq -1$. Furthermore,  for any $a,b \in I_1$, if $\gamma(a) = \gamma(b) \neq -1$, then $a=b$. Since the domain of $\gamma$ is $I_1$, then $\gamma^{-1}(t) \in I_1$. This means that $\gamma^{-1}(t)$ is in the domain of $F(t)$ and thus one has that $F'(t)$ is defined for all $t\in I_2$.

We shall now  prove the result in the lemma by induction. Since $\gamma^{-1}(0)=0$ and the initially infected sets for the models with $F(t)$ and $F'(t)$ are the same, it must be true that $A^{F' }_{0} \subseteq A^{F }_{0}$, and in particular, $A^{F' }_{0} = A^{F }_{0} = A_0.$ In order to perform the inductive step, suppose that for some $t \in I_2$ and $t+1 \in I_2$, one has $A^{F' }_{t} \subseteq A^{F }_{\gamma^{-1}(t)}$. Moreover, suppose there is a node $n$ such that $n \in A^{F' }_{t+1}$ but $n \notin A^{F }_{\gamma^{-1}(t+1)}$. Then, this means that there exists a neighbor $n'$ of $n$ such that $n' \in A^{F' }_{t}$ but $n' \notin A^{F }_{\gamma^{-1}(t+1)-1}$. Indeed, otherwise this would imply  that the set of neighbors of $n$ infected prior to the specified times are the same for both models, and since $F'(t+1) = F(\gamma^{-1}(t+1))$ for $t \in I_2$, and thus $n$ would be infected in both or neither models. 
From the above, since $t < t+1$ one must have $\gamma^{-1}(t) < \gamma^{-1}(t+1)$, and thus   $$\gamma^{-1}(t) \leq \gamma^{-1}(t+1)-1.$$ Moreover, since $n' \notin A^{F }_{\gamma^{-1}(t+1)-1}$, then $n' \notin A^{F }_{\gamma^{-1}(t)}$. However, we assumed $n' \in A^{F' }_{t}$, and since $A^{F' }_{0} \subseteq A^{F }_{0}$, we have a contradiction, so it must be true that the sets satisfy $A^{F' }_{t+1} \subseteq A^{F }_{\gamma^{-1}(t+1)}$. Thus we have proven that for any initially infected set $A_0$ and $t \in I_2$, one has that \eqref{mas11} is satisfied for all $t\in I_2$.
\end{proof}
 Through the above lemma we can further understand when an $F(t)$-percolation process finishes in the following manner. 
\begin{lemma}
Given a percolation function   $F(t)$ and a fixed time $t \in \mathbb{N}$, let $t_p<t$ be such that $F(t_p) < F(t)$, and suppose there does not exist another time $t_i \in \mathbb{N}$ where $t_p < t_i <t$ such that $F(t_i) < F(t)$. Suppose further that we use this percolation function on a graph with $\ell$ vertices. Then, if  $|\{t_i~|~F(t_i)=F(t)\}|>\ell$, then there are no nodes that becomes infected at time $t$.
\end{lemma}
\begin{proof}
Suppose some node $n$ is infected at time $t$. 
Then, this would imply that all nodes are infected before time $t$. We can show this using contradiction: suppose there exists $m$   nodes $n_i$ that there are not infected by time $t$. Then we know that there exists at least $m$ of $t_j \in \mathbb{N}$ such that $t_p < t_j < t$, for which $F(t_j) = F(t)$ and such that there is no node infected at $t_j$. Matching each $n_i$ with some $t_j$ and letting $t_k \in \mathbb{N}$ be such that $t_j < t_k \leq t$, one can see that there is some node infected at $t_k$, and $F(t_k) = F(t)$. Moreover, this implies that there is no $t_x \in \mathbb{N}$ such that $t_j < t_x < t_k$ and such that there is some node infected at $t_x$ and $F(t_x) = a$. We know such a $t_k$ exists because there is a node infected at time $t$. 

From the above, for each $n_i$ there are two cases: either the set of nodes infected by $t_j$ is the same as the set of nodes infected by $t_k$, or there exists node $p$ in the set of nodes infected by $t_k$ but not in the set of nodes infected by its $t_j$. We have a contradiction for the first case: there must be a node infected at time $t_j$ is this is the case, as the set of infected nodes are the same as time $t_k$, so the first case is not possible. So the second case must hold for all $m$ of $n_i$'s. But then, the second case implies that there is a node infected between $t_j$ and $t_k$. This means that at least $m$ additional nodes are infected, adding to the at least $\ell-m$ nodes infected at $t_i$ such that $F(t_i) = a$ and there is a node infected at $t_i$, we have at least $\ell-m+m=\ell$ nodes infected before $t$.
But if all $\ell$ nodes are infected before $t$, this would mean there are no nodes to infect at time $t$, so $n$ does not exist.
\end{proof}

 Intuitively, the above lemma tells us that given a fixed time $t_0$ and some $t>t_0$, if $F(t) = \ell$ is the smallest value the function takes on after the time $t_0$, and $F(t)$ has already taken on that value more than $\ell$ times, for $\ell$ the number of nodes in the graph,  then there will be no nodes that will be infected at that time and the value is safe to be ``removed''. The removal process is clarified in the next proposition, where we  define an upper bound of percolation time on a specified tree and function $F(t)$.

\begin{prop} \label{propo8}Let  $G$ be a  regular tree of degree $d$  and $\ell$ vertices. Given a percolation function $F(t)$,  define the functions $F'(t)$ and $\gamma: \mathbb{N} \rightarrow \mathbb{N} \cup \{-1\}$ by setting:
\begin{itemize}
\item[(i)] $F'(0) := F(0)$, and $\gamma(0) := 0$.
\item[(ii)] Suppose the least value we have not considered $F(t)$ at is $a$, and let $b$ be the least value where $F'(b)$ has not yet been defined. If $F(a)$ has not yet appeared $\ell$ times since the last time $t$ such that $F(t) < F(a)$ and $F(a) \leq d$, then set $F'(b) := F(a)$, and let  $\gamma(a)=b$. Otherwise, $\gamma(a)=-1$.\end{itemize}
The function $F'(t)$ is equivalent to $F(t)$. \label{P1}
\end{prop}
\begin{proof}
Intuitively,  the function $\gamma$ constructed above is mapping the index associated to $F(t)$ to the index associated to $F'(t)$. If omitted, then it is mapped to $-1$ by $\gamma$.
To prove the proposition, we will prove that $P_{F(t)}(A) = P_{F'(t)}(A)$. Suppose we have a node $n$ in $P_{F(t)}(A)$, and it is infected at time $t_0$. Suppose $F(t_0) = a$ for some $a \in \mathbb{Z}^+$, and let $t_{prev}$ be     the largest integer $t_{prev} < a$ such that $F(t_{prev}) < a$. Suppose further that $t_0$ is the $m$th instance such that $F(t) = a$ for some $t$. Moreover, if $m > v$, there cannot be any node infected at time $t_0$ under $F(t)$, and thus it follows that $m \leq v$. But if $m \leq v$, then $\gamma(t_0) \neq -1$ and therefore  all nodes that are infected under $F(t)$ became infected at some time $t$ where $\gamma(t_0) \neq -1$. 

Recall that  $A_0^{F} = A_0^{F'}$, and    suppose for some $n$ such that $\gamma(n)\neq -1$, one has that $A_n^{F} = A_{\gamma(n)}^{F'}$. We know that for any $n < t < \gamma^{-1}(\gamma(n)+1), \gamma(t) = -1$, so nothing would be infected under $F(t)$ after time $n$ but before $\gamma^{-1}(\gamma(n)+1)$. This means that the set of previously infected nodes at time  $\gamma^{-1}(\gamma(n)+1)-1$ is the same as the set of nodes infected before time $n$ leading to \[A_n^{F} = A_{\gamma^{-1}(\gamma(n)+1)-1}^{F'}.\] Then, since $F(\gamma^{-1}(\gamma(n)+1)) = F'(\gamma(n)+1)$ and the set of previously infected nodes for both are $A_n^{F}$, we know that $A_{n+1}^{F} = A_{\gamma(n+1)}^{F'}$. Thus, for any time $n'$ in the domain of $F'(t)$, there exist a corresponding time $n$ for percolation under $F(t)$ such that the infected set at time $n$ under $F(t)$ and the infected set at time $n'$ under $F'(t)$ are the same, and thus $A_\infty^{F} = A_\infty^{F'}$.
\end{proof}

From  the above Proposition \ref{P1} we can see two things: the upper bound on the percolation time is the time of the largest $t$ such that $F'(t)$ is defined, and we can use this function in an algorithm to find the smallest minimal percolating set since $F(t)$ and $F'(t)$ are equivalent.
Moreover, an upper bound on the percolation time can not be obtained without regards to the percolation function: suppose we have such an upper bound $b$ on some connected graph with degree $d$ and with $1$ node initially infected and more than $1$ node not initially infected. Then, if we have percolation function $F(t)$ such that $F(t) = d+1$ for all $t\in \mathbb{N} \leq b$ and $F(m)=1$ otherwise, we see that there will be nodes infected at time $b+1$, leading to a contradiction.

\begin{lemma} 
Suppose the degree of a graph is $d$. Define a sequence $a$ where $a_1 = d$ and $a_{n+1} = (a_n+1)d$. Then the size of the domain of $F'(t)$ in Proposition \ref{P1} is 
$\Sigma^{d}_{i=1}a_n$. \label{ll}
\end{lemma}
\begin{proof}
Suppose each value do appear exactly $d$ times after the last value smaller than it appears. To count how large the domain can be, we start with the possible $t$s such as $F'(t)=1$s in the function; there are $d$ of them as $1$ can maximally appear $d$ times. Note that this is equal to $a_1$. Now, suppose we have already counted all the possible $t$s when $F'(t) < n+1$, for $1 leq n < d$, which amounted to $a_{n}$. Then, there can be maximally $d$ instances at the between the appearance of each $t$ when $F'(t) < n$ as well as before and after all such appearances, so there are $a_{n}+1$ places where $F'(t)=n$ can appear. Thus there are maximally $(a_{n}+1)d$ elements $t$ in the domain such that $F'(t) = n+1$. Summing all of them yields $\Sigma^{d}_{i=1}a_n$, the total number of elements in the domain.
\end{proof}

\begin{remark}
From Proposition \ref{P1}, for some $F(t)$, $A_0$ and $n$, one has $A^{F}_{\gamma^{-1}(n)} = A^{F'}_{n}$. Then if $A_\infty^{F'}$ is reached by time $\Sigma^{d}_{i=1}a_n$, the set  must  be infected by  time $\gamma^{-1}(\Sigma^{d}_{i=1}a_n)$. Hence, in this setting an upper bound of $F(t)$ percolating on a graph with $d$ vertices can be found by  taking $\gamma^{-1}(\Sigma^{d}_{i=1}a_n)$, as defined in Lemma \ref{ll}.
\end{remark}

\section{Minimal Percolating Sets}\label{minimal}
 
  When considering percolations within a graph, it is of much interest to understand which subsets of vertices, when infected, would lead to the infection reaching the whole graph. 
 \begin{definition}
 A {\em percolating set} of a graph $G$ with percolation function $F(t)$ is a set $A_0$  for which   $A_\infty^F=G$  at a finite time. A {\em minimal percolating set} is a percolating set $A$ such that if    any node  is removed from $A$, it will no longer be a percolating set. 
\end{definition}

\begin{remark}A natural motivation for studying minimal percolating sets is that as long as we keep the number of individuals infected to less than the size of the minimal percolating set, we know that the entire population will not be decimated. 
\end{remark}
Bounds on minimal percolating sets on grids and other less regular graphs have extensively been studied.  For instance,  it has been shown  in \cite{Morris} that for a grid $[n]^d$, there exists a minimal percolating set of size $4n^2/33 + o(n^2)$, but there does not exist one larger than $(n + 2)^2/6$. In the case of trees, \cite{percset} gives an algorithm that finds the largest and smallest minimal percolating sets on trees. However, the results in the above papers cannot be easily extended to the dynamical model because it makes several assumptions such as $F(t) \neq 1$ that do not necessarily hold in the dynamical model.

\begin{example}\label{ex2}An example of a minimal percolating set with $F(t)=t$ can be seen in Figure \ref{ex1} (a). In this case, the minimal percolating set has size 3. Indeed, we see that if we take away any of the red nodes, the remaining initially infected red nodes would not percolate to the whole tree, and thus they form a minimal percolating set; further, there exists no minimal percolating sets of size 1 or 2, thus this is the smallest minimal percolating set.  It should be noted that minimal percolating sets can have different sizes. For example, another minimal percolating set with $5$ vertices appears in  Figure \ref{ex1} (b). 

\begin{figure}[!h]
\centering
\includegraphics[scale=.15]{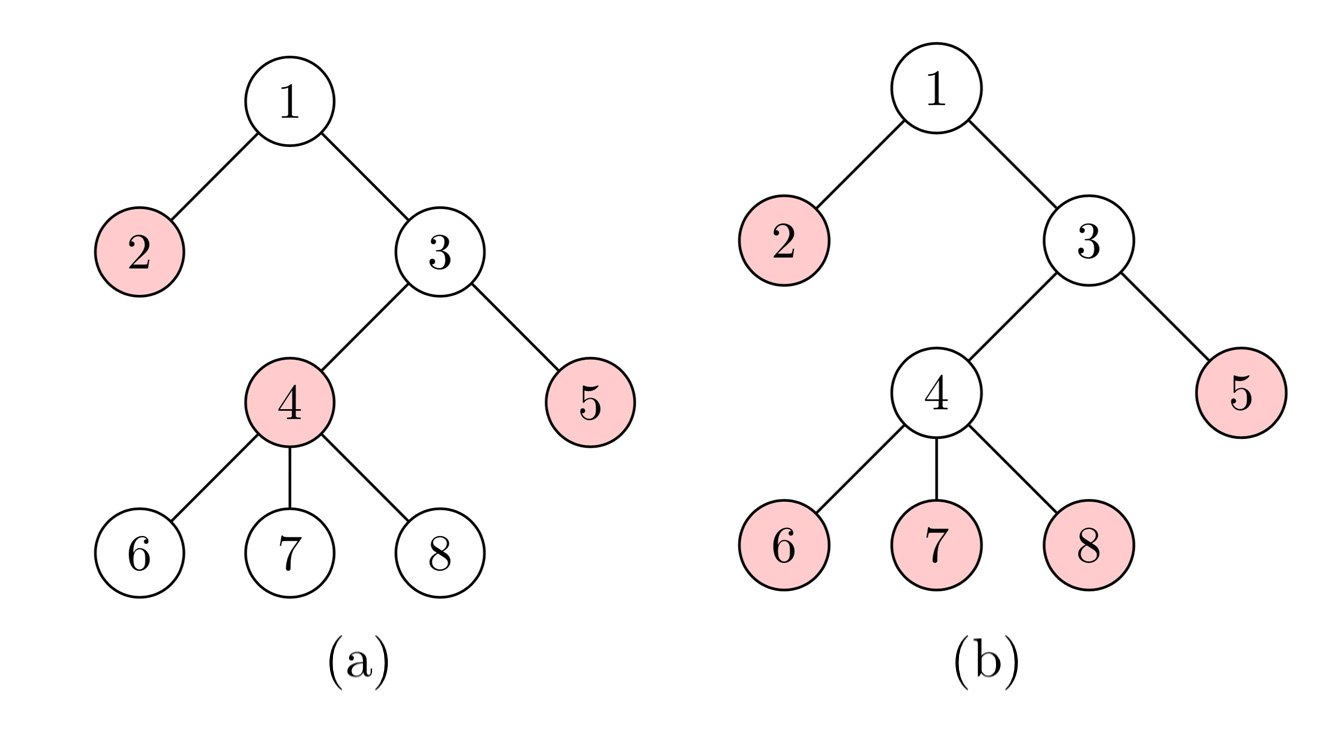}
\caption{(a) In this tree, having nodes $2,4,5$ infected (shaded in red) initially is sufficient to ensure that the whole tree is infected. (b) This minimal percolating set shaded in red is of size $5$.}\label{ex1}
\end{figure}

\end{example}

In what  follows we shall work with general finite trees $T(V,E)$ with set of vertices  $V$  and  set of edges $E$. In particular, we shall consider the smallest minimal percolating sets in the following section.

\section{Algorithms for Finding Smallest Minimal Percolating Set}\label{minimal2}

Consider $F(t)$-bootstrap percolation on a tree $T(V,E)$ with initially infected set $A_0\subset V$. As before, we shall denote by   $A_t$ be the set of nodes infected at time $t$.  
 For simplicity, we shall  use here the word ``infected'' synonymously with ``infected''. In order to build an algorithm to find smallest percolating sets, we first need to introduce a few definitions that will simplify the notation at later stages.

\begin{definition}
We shall denote by $L(a)$  the largest time $t$ such that $a \leq F(t),$ and if there does not exist such a time $t$, then set $L(a)=\infty$.  Similarly,  define $B(a)$ as  the smallest  time $t$ such that $a \leq F(t)$, and if such a time $t$ does not exist, set $B(a)=\infty$.
\end{definition}

Given  $a,b\in \mathbb{N}$, if   $a<b$ then $L(a) \geq L(b)$. 
 Indeed, this holds   because if a node can be infected to with $b$ neighbors, it can with $a$ neighbors where $a<b$. Note that in general, a smallest percolating set $A_0$ must be a minimal percolating set. To see this, suppose not. Then there exists some $v$ in $A_0$ such that $A_0 -\{v\}$ percolates the graph. That means that $A_0 -\{v\}$, a smaller set that $A_0$, is a percolating set. However, since $A_0$ is a smallest percolating set, we have a contradiction.
 Hence, showing that a percolating set $A_0$ is the smallest implies that $A_0$ is a minimal percolating set.

\begin{remark}
The first algorithm that comes to mind is to try every case. There are $2^n$ possible sets $A_0$, and for each set we much percolate $A_0$ on $T$ to find the smallest percolating set. This amounts to an algorithm of complexity $O(t2^n)$ where $t$ is the upper bound on the percolation time.
\end{remark}
In what follows we shall describe a polynomial-timed algorithm to find the smallest minimal percolating set on $T(V,E)$, described in  Theorem \ref{teorema}. For this, we shall introduce two particular times associated to each vertex in the graph, and formally define what isolated vertices are.

\begin{definition}
For each node $v$  in the graph, we let  $t_a(v)$ be the time when it is  infected, and $t_*(v)$ the time when it is last allowed to be infected; 
\end{definition}Moreover, when building our algorithm, each vertex will be allocated  a truth value of whether it needs to be further considered.

\begin{definition}
A node $v$ is said to be {\em isolated} with regards to $A_0$ if there is no vertex $w\in V$  such that $v$ becomes infected when considering $F(t)$-bootstrap percolation with initial set $A_0 \cup \{w\}$. 
\end{definition}

From the above definition, a node is isolated with regards to a set if it is impossible to infect it by adding one of any other node to that set that is not itself.
Building towards the percolating algorithm, we shall consider a few lemmas first.

\begin{remark}
If a node cannot be infected by including a neighbor in the initial set, it is isolated.
\label{L1} \end{remark}

From Remark \ref{L1}, by filling the neighbor in the initial set, we either increased the number of neighbors infected to a sufficient amount, or we expanded the time allowed to percolate with fewer neighbors so that percolation is possible. We explore these more precisely in the next lemma, which gives a quick test to see whether a vertex is isolated.

\begin{lemma} \label{L3}
Let $v$ be an uninfected node such that not all of its $n$ neighbors are in set $A_0$. Define function \begin{eqnarray}
N:\{0,1,...,n\} \rightarrow \mathbb{Z}\label{NN}\end{eqnarray} where $N(i)$ is the smallest time when $i$ of the neighbors of node $v$ is infected, and set $N(0)=0$. Then, a vertex $v$ is isolated iff there exists no $i$ such that \[F(t) \leq i+1~ {\rm for~ some~} t \in (N(i),t_*].\]
\end{lemma}
\begin{proof}
Suppose $s\in N(v)\cap A_0$. Then, if there exists $i$ such that $F(t) \leq i+1$ for some $t \in (N(i),t_*]$, using $A_0 \cup \{s\}$ as the initially infected set allows percolation to happen at time $t$ since there would be $i+1$ neighbors infected at each time $N(i)$. Thus with contrapositive, the forward direction is proven.

Let $v$ be not isolated, and $v \in P(A_0 \cup \{s\})$ for some neighbor $s$ of $v$. Then there would be $i+1$ neighbors infected at each time $N(i)$. Moreover, for $v$ being to be infected, the $i+1$ neighbors must be able to fill $v$ in the allowed time, $(N(i),t_*]$. Thus there exists $N(i)$ such that $F(t) \leq i+1$ for some $t \in (N(i),t_*]$. With contrapositive, we proved the backwards direction.
\end{proof} 

Note that if a vertex $v$ is uninfected and $N(v)\subset A_0$, then the vertex must be isolated.
In what follows we shall study the effect of having different initially infected sets when studying $F(t)$-bootstrap percolation.
\begin{lemma}\label{L2} Let $Q$ be an initial set for which a fixed vertex $v$ with $n$ neighbours  is isolated. 
Denoting the neighbors of $v$ be $s_1, s_2,...,s_n$, we let the   times at which they are infected  be $t_1^Q, t_2^Q,\ldots,t_n^Q$. Here,  if for some $1\leq i \leq n$, the vertex $s_i$ is not infected, then set $t_i^Q$ to be some arbitrarily large number.
 Moreover, consider  another initial set $P$  such that the times at which $s_1, s_2,..., s_n$ are infected are $t_1^P, t_2^P,\ldots,t_n^P$  satisfying  
 \begin{eqnarray}t_i^Q=&t_i^P&~{\rm for }~ i\neq j;\nonumber\\
 t_j^Q \leq& t_j^P&~{\rm for }~ i= j,\nonumber
 \end{eqnarray}
 for some $1 \leq j \leq n$. If $v \notin P$, then the vertex $v$ must be isolated with regards to $P$ as well.
\end{lemma}
\begin{proof}
Consider $N_Q(i)$  as defined in \eqref{NN} for the set $Q$,  and $N_P(i)$  the corresponding function for the set $P$. Then it must be true that for all $k \in \{0,1,...,n\}$, one has that $N_Q(k) \leq N_P(k)$. Indeed,  this is because with set $P$, each neighbor of $v$ is infected at or after they are with set $Q$. Then, from  Lemma \ref{L3}, $v$ is isolated with regards to $Q$ so there is no $m$ such that 
\[F(t) \leq m+1{~\rm~ for~ some~ }~t \in (N_Q(m),t_*].\] However, since \[N_Q(k) \leq N_P(k){~\rm~ for~ all~ }~k \in \{0,1,...,n\},\] we can say that there is no $m$ such that \[F(t) \leq m+1{~\rm~ for~ some~ }~t \in (N_P(m),t_*]\] as $(N_P(m),t_*] \subseteq (N_Q(m),t_*].$ Thus we know that $v$ must also be isolated with regards to $P$.
\end{proof}

\begin{definition} \label{D2}
Given a vertex $v$ which is not isolated, we define   $t_p(v)\in (0,t_*]$ to be be the largest integer  such that there exists $N(i)$ where $F(t_p) \leq i+1$.
\end{definition}

Note that in order to fill an isolated node $v$, one can fill it by filling one of its neighbors by time $t_p(v)$, or just add the vertex it to the initial set. Hence, one needs to   fill a node $v_n$ which is either the parent ${\rm par}(v_n)$, a child ${\rm chi}(v_n)$, or itself.

\begin{lemma}
Let $v\notin A_0$ be an isolated node $v$. 
To achieve percolation, it is always better (faster) to include $v$ in $A_0$ than attempting to make $v$ unisolated.
\end{lemma}
\begin{proof}
It is possible to make $v$ isolated by including only descendants of $v$ in $A_0$ since we must include less than $deg(v)$ neighbors. But we know that if given the choice to include a descendant or a $v$ to the initial set, choosing $v$ is absolutely advantageous because the upwards percolation achieved by $v$ infected at some positive time   is a subset of upwards percolation achieved by filling it at time $0$. Thus including $v$ to the initial set is superior.
\end{proof}
The above set up can be understood further to find which vertex needs to be chosen to be $v_n$.
\begin{lemma} Consider a vertex  $v\notin A_0$. Then, in finding a node $u$ to add to $A_0$ so that $v \in A_\infty$ for the initial set $A_0 \cup \{u\}$ and  such $A_\infty$ is maximized, the vertex $v_n$ must be the parent ${\rm par}(v)$ of $v$.
\end{lemma}
\begin{proof}
Filling $v$ by time $t_*(v)$ already ensures that all descendants of $v$ will be infected, and that all percolation upwards must go through the parent ${\rm par}(v)$ of $v$. This means that filling any child of $v$ in order to fill $v$ (by including some descendant of $v$ in $A_0$) we obtain a subset of percolation if we include the parent ${\rm par}(v)$ of $v$ in $A_0$. Therefore,   the parent ${\rm par}(v)$ of $v$ or a further ancestor needs to be included in $A_0$, which means $v_n$ needs to be the parent ${\rm par}(v)$ of $v$.
\end{proof}

Note that given a node $v\notin A_0$,  if we fill its parent ${\rm par}(v)$ before $t_p(v)$, then the vertex will be infected. We are now ready for  our main result, which improves   the naive $O(t2^n)$ bound for finding minimal percolating sets to $O(tn)$, as discussed further in the last section.

\begin{theorem}\label{teorema}\label{teo1}
To obtain one smallest minimal percolating set of a tree $T(V,E)$ with percolation function $F(t)$, proceed as follows:
\begin{itemize}
\item Step 1. initialize tree: for each node $v$, set $t_*(v)$ to be some arbitrarily large number, and set it to true for needing to be considered. 

\item Step 2. percolate using current $A_0$. Save the time $t_a$'s at which the nodes were infected. Stop the algorithm if the set of nodes that are infected equals the set $V$.

\item Step 3. consider a node $v$ that is furthest away from the root, and if there are multiple such nodes, choose the one that is isolated, if it exists. 

\begin{itemize}
\item if $v$ is isolated or is the root, add $v$ to $A_0$.
\item otherwise, set $t_*({\rm par}(v))=t_p(v)-1$  (as Definition \ref{D2})  if it is smaller than the current $t_*({\rm par}(v))$ of the parent.
\end{itemize}
Set $v$ as considered.

\item Step 4. go to step 2.
\end{itemize}
After the process has finished, the resulting set $A_0$ is one of the smallest minimal percolating set.
\end{theorem}
\begin{proof}
The proof of the  theorem, describing the algorithm through which one can find a smallest percolating set,  shall be organized as follows: we will first show that the set $A_0$ constructed through the steps of the theorem is a minimal percolating set, and then show that it is the smallest such set.
In order to see that $A_0$ is a minimal percolating set, we first need to show that $A_0$ percolates. In step 3, we have included all isolated nodes, as well as the root if it wasn't infected already, in $A_0$ and guaranteed to fill all other nodes by guaranteeing that their parents will be infected by their time $t_p$.

Showing that $A_0$ is a minimal percolating set is equivalent to showing that if we remove any node from $A_0$, it will not percolate to the whole tree. Note that in the process, we have only included isolated nodes in $A_0$ other than the root. This means that if any node $v_0$ is removed from $A_0$, it will not percolate to $v_0$ because we only fill nodes higher than $v_0$ after considering $v_0$ and since turning a node isolated requires filling at least one node higher and one descendant of $v_0$, it cannot be infected to after removing it from $A_0$. Moreover, if the root is in $A_0$, since we considered the root last, it is implied that the rest of $A_0$ does not percolate to root. Thus, $A_0$ is a minimal percolating set.

Now we show that the set $A_0$ constructed through the algorithm is of the smallest percolating size by contradiction using Lemma \ref{L2}. For this, suppose there is some other minimal percolating set $B$ for which $|B|\leq |A|$. Then, we can build an injection $A_0$ to $B$ in the following manner: iteratively consider the node $a$ that is furthest from the root and $a \in A_0$ that hasn't been considered, and map it to a vertex $b_0$ which is  itself or one of its descendants of $b$ where $b \in B$. We know that such a $b_0$ must exist by induction.

We first consider the case where $a$ has no descendant in $A$.  Then, if  the vertex $b\in B$ and $b$ is a descendant of $a$, we map $a$ to $b$. Now suppose there is no node $b$ that is a descendant of $a$ where $b \in B$. Then, $a \in B$ because otherwise $a$ would be isolated with regards to $B$ as well, by Lemma \ref{L2}. This means that we can map $a$ to $a$ in this case.

Now we can consider the case where all the descendants $d$ of $a$ such that $d \in A:=A_0$ has been mapped to a node $b_d\in B$ where $b_d$ is $d$ or a descendant of $d$. If there is such a $b\in B$, then $b$ is a descendant of $a$, and thus no nodes in $A$ have  been matched to $b$ yet,  allowing us to map $a$ to $b$. Now suppose there is no such $b\in B$. This means that there is no $b\in B$ such that all of the descendants of $a$ are descendants of $b$. Then, all nodes in $B$ that are descendants of $a$ is either some descendant of $a\in A$ or some descendant of a descendant of $a$ in $A$. This means that percolating $B$, the children of $a$ will all be infected at later times than when percolating $A$, and by Lemma \ref{L2}, one has that $a \in B$ because $a$ would be isolated with regards to $B$. So in this case, we can map $a$ to $a$.

The map constructed above is injective because each element of $B$ has been mapped to not more than once.
Since we constructed an injective function from the set generated by the algorithm $A_0$ to a smaller minimal percolating set $B_0$, we have a contradiction because $A_0$ then must be the same size or larger than $B_0$. Thus, the set generated from the algorithm must be a smallest minimal percolating set.
\end{proof}

From Theorem \ref{teo1}   one can  find the smallest minimal percolating set on any finite tree. Moreover, it gives an intuition for how to think of the vertices of the graph: in particular, the property of ``isolated'' is not an absolute property, but a property relative to the set of nodes that has been infected before it. This isolatedness is easy to define and work with in trees since each node has at most one parent. Moreover, a similar property may be considered in more general graphs and we hope to explore this in future work. Below we shall demonstrate the algorithm of Theorem \ref{teo1} with an example.

\begin{example}
We will preform the algorithm on the tree in Example \ref{ex2}, with percolating function $F(t)=t$. We first initialize all the nodes, setting their time $t_*$ to some arbitrarily large number, represented as $\infty$ in Figure \ref{inf1} below. 
\begin{figure}[h!]
\centering
 \centering
\includegraphics[scale=.25]{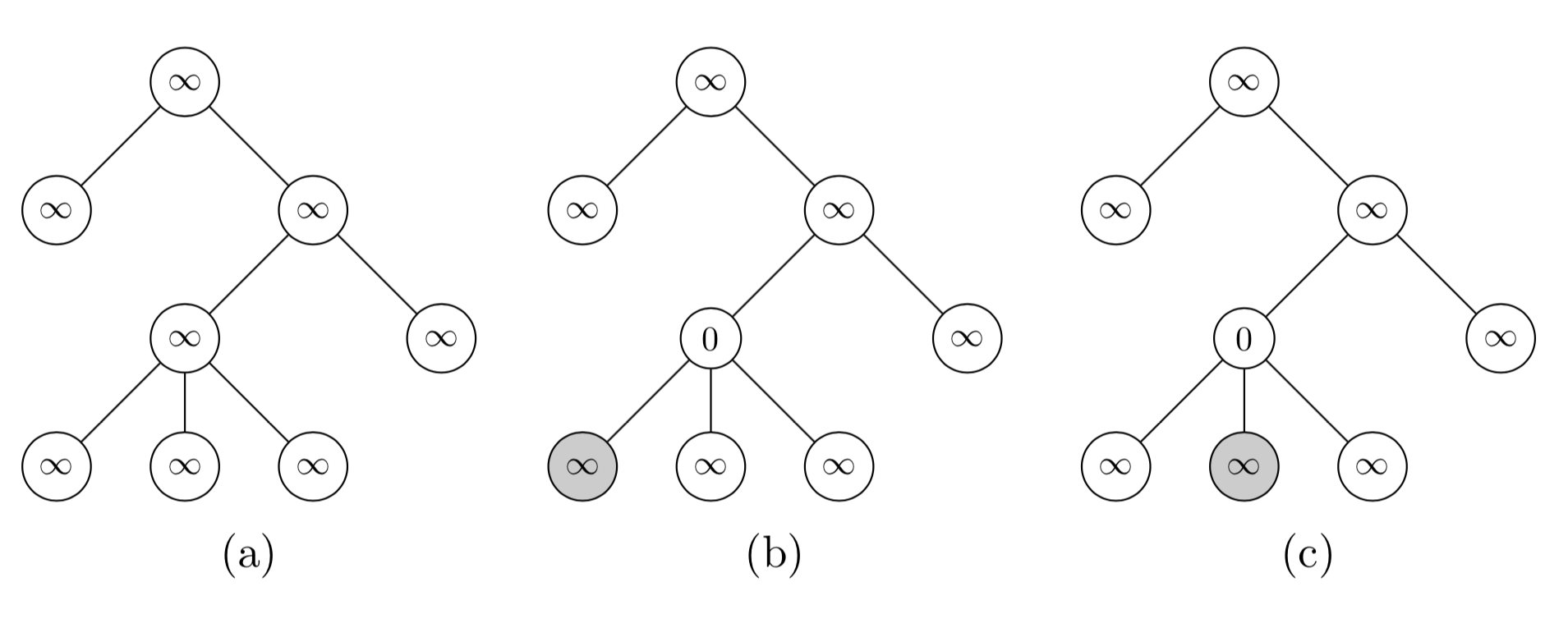}
\caption{(a)-(c) show the first three updates through the algorithm in Theorem \ref{teo1}, where the vertices considered at each time are shaded and each vertex is assigned the value of $t_*$. }\label{inf1}
\end{figure}

Percolating the empty set $A_0$,    the resulting infected set is empty, as shown in Figure \ref{inf1} (a). We then consider the furthest node from root. None of them are isolated, so we can consider any; we begin by considering node $6$ in the labelling of Figure \ref{ex1} of Example \ref{ex2}. It is not isolated, so we set the $t_*$ of the parent to $t_p-1=0$, as can be seen in Figure \ref{inf1}  (b). Then we consider another node furthest from the root, and through the algorithm set the $t_*$ of the parent to $t_p-1=0$, as can be seen in Figure \ref{inf1}  (c). The following steps of the algorithm are depicted in Figure \ref{inf2} below.

\begin{figure}[h!]
\centering
 \centering
\includegraphics[scale=.25]{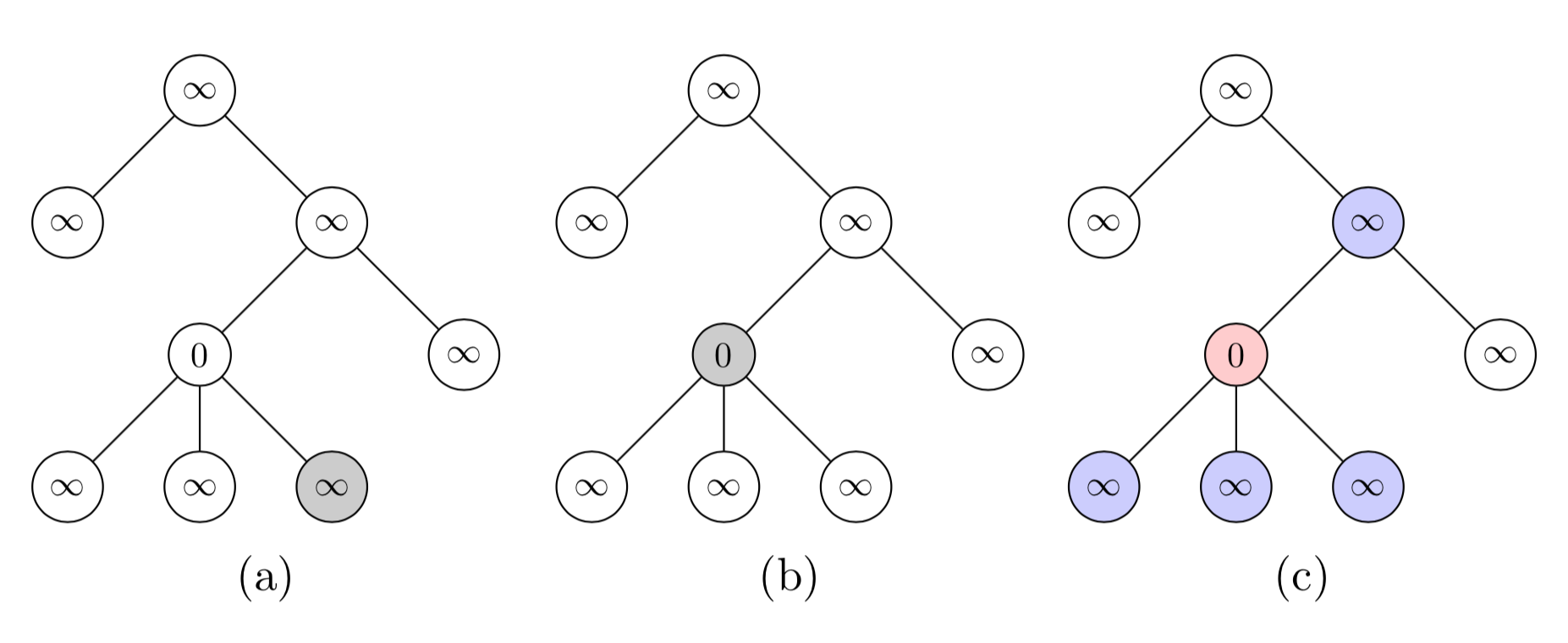}
\caption{ (a)-(b) show the updates 4-5 through the algorithm.   (c) shows the set $A_0$ in   red, and the infected vertices in blue.
 }\label{inf2}
\end{figure}

As done in the first three steps of Figure \ref{inf1}, we consider the next furthest node $v$ from the root, and by the same reasoning as node $6$, set the $t_*{\rm par}(v)$ of the parent to $t_*{\rm par}(v)=1$, as can be seen in Figure  \ref{inf2} (a). Now we consider node $4$: since it is isolated, so we fill it in as in Figure \ref{inf2} (b). The set of nodes infected can be seen in Figure \ref{inf2} (c). We then consider node $5$, the furthest node from the root not considered yet. Since it is not isolated,  change the $t_*{\rm par}(v)$ of its parent to $t_p(v)-1=0$, as   in  Figure \ref{inf3} (a). 

\begin{figure}[h!]
\centering
\centering
\includegraphics[scale=.25]{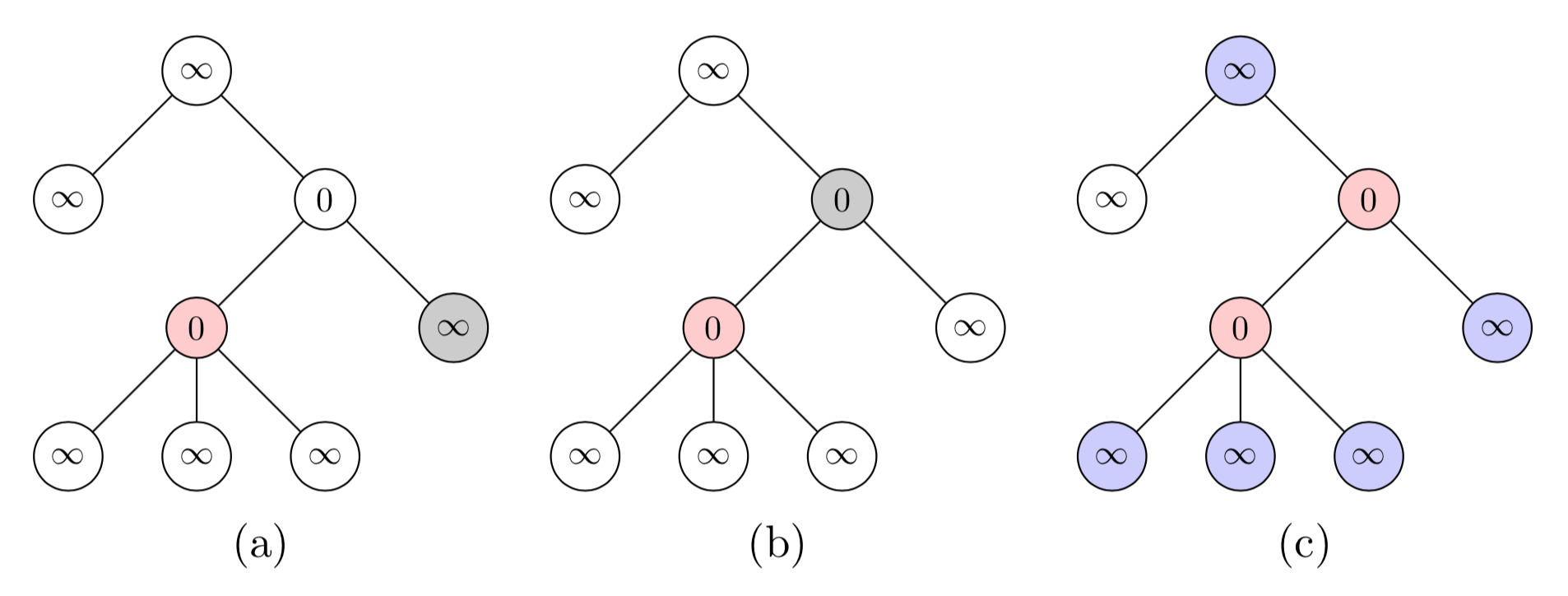}
\caption{(a)-(c) show the updates  through the algorithm in Theorem \ref{teo1} after setting $A_0$ to be as in Figure \ref{inf2}.}\label{inf3}
\end{figure}

Then we consider node $3$, which is isolated, so we include it in $A_0$. The infected nodes as a result of percolation by this $A_0$ is shown as red vertices in Figure \ref{inf3} (c).
In order to finish the process, consider the vertex $v=2$ since it is the furthest away non-considered node. It is not isolated so we change the $t_*{\rm par}(v)$ of its parent to $t_p(v)-1=0$, as shown in Figure  \ref{inf4}  (a). Finally, we consider the root: since it is isolated, we include it in our $A_0$ as seen in Figure  \ref{inf4} (b). Finally, percolating this $A_0$ results in all nodes being infected as shown in Figure  \ref{inf4}  (c), and thus we stop our algorithm.

\begin{figure}[H]
\centering
 \centering
\includegraphics[scale=.28]{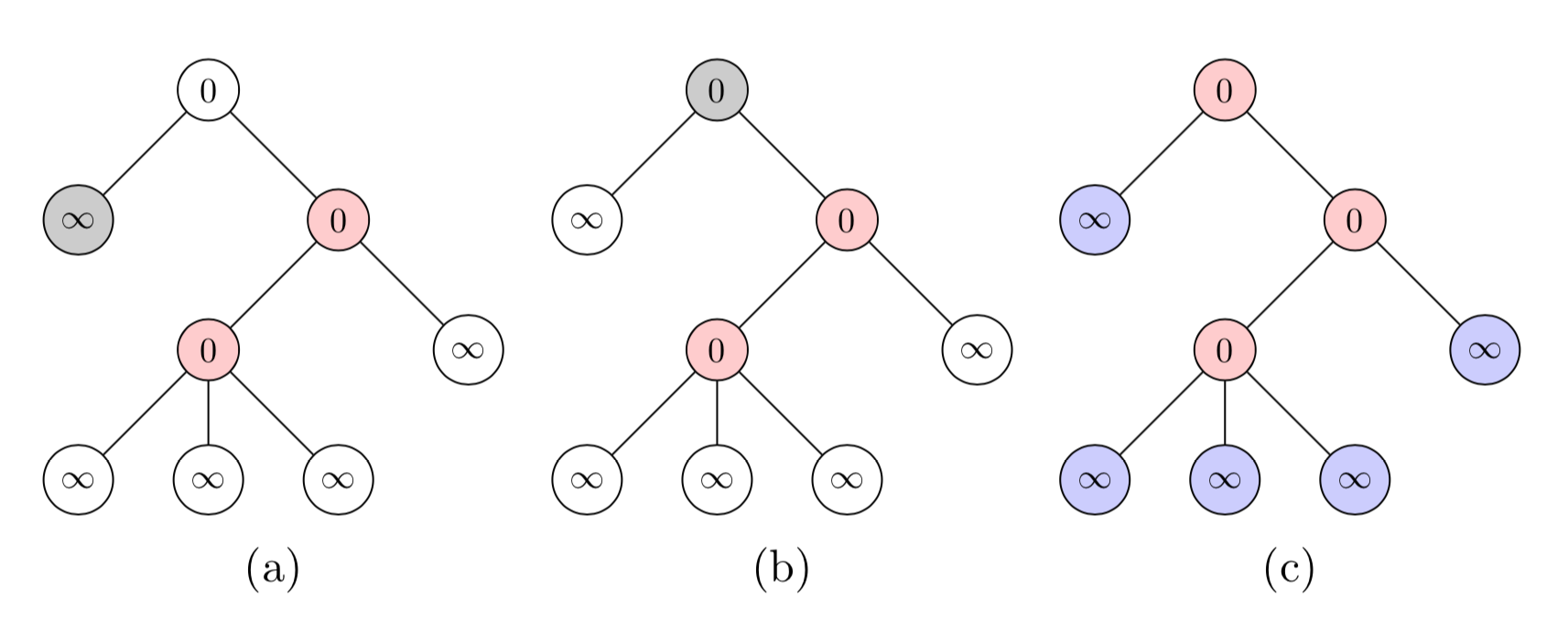}
\caption{Final steps of the algorithm.}\label{inf4}
\end{figure}

Through the above algorithm, we have constructed a smallest minimal percolating set shown as red vertices in Figure  \ref{inf4} (c), which is of size $3$. Comparing it with Example \ref{ex2}, we see that the minimal percolating set in that example is indeed the smallest, also with $3$ elements. Finally, it should be noted that  in general the times  $t_p$ for each node could be different from each other and are not the same object. 
\end{example}
From the above example, and its comparison with Example \ref{ex2}, one can see  that a graph can have multiple different smallest minimal percolating sets, and the algorithm finds just one.
In the algorithm of Theorem \ref{teo1}, one  minimizes the size of a minimal percolating set ,\ relying on the fact that as long as a node is not isolated, one can engineer its parent to become infected so as to infect the initial node. The motivation of the definition of isolated stems from trying to find a variable that describes whether a node is still possible to become infected by infecting its parent. Because the algorithm is on trees, we could define isolation to be the inability to be infected if we add only one node.

\section{Concluding remarks}\label{final}
In order to show the relevance of our work, we shall conclude this note with a short comparison of our model   with those existing in the literature. \\

\noindent{\bf Complexity.} Firstly we shall consider the  complexity of the algorithm in Theorem \ref{teo1} to find the smallest minimal percolating set on a graph with $n$ vertices. To calculate this, suppose $t$ is the upper bound on percolation time; we have presented a way to find such an upper bound in the previous sections. 
In the algorithm, we first initialize the tree, which is linear timed. Steps $2$ and $3$ are run at most $n$ times as there can only be a total of $n$ unconsidered nodes. The upper bound on time is $t$, so steps 2 will take $t$ to run. Determining whether a node is isolated is linear timed, so determining isolated-ness of all nodes on the same level is quadratic timed, and doing the specifics of step 3 is constant timed. Thus the algorithm is $O(n+n(t+n^2)) = O(tn + n^3) = O(tn)$, much better than then $O(t2^n)$ complexity of the naive algorithm.\\

\noindent{\bf Comparison on perfect trees.}  Finally,  we shall compare our algorithm with  classical $r$-bootstrap percolation. For this, in Figure \ref{comp} we show a comparison of sizes of the smallest minimal percolating sets on perfect trees of height $4$, varying the degree of the tree. Two different functions were compared: one is constant and the other is quadratic. We see that the time-dependent bootstrap percolation model can be superior in modelling diseases with time-variant speed of spread, for that if each individual has around $10$ social connections, the smallest number of individuals needed to be infected in order to percolate   the whole population has a difference of around $10^3$ between the two models.

\begin{figure}[H]
    \centering
        \resizebox{0.45\textwidth}{!}{%
    \includegraphics{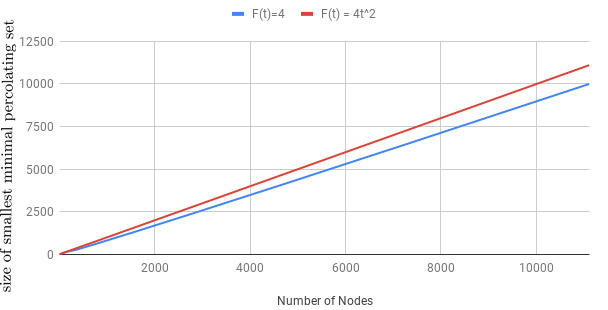}}
    \caption{The size of smallest minimal percolating sets on perfect trees with height 4, with a constant and a non-constant percolation function $F(t)$.}
    \label{comp}
\end{figure}
\smallbreak

\noindent{\bf Comparison on random trees.}  We shall conclude this work by comparing the smallest minimal percolating sets found through our algorithm and those constructed by Riedl in  \cite{percset}. In order to understand the difference of the two models, we shall first consider in Figure \ref{comp1} three percolating functions $F(t)$ on random trees of different sizes, where each random tree has been formed by beginning with one node, and then for each new node $i$ we add, use a random number from $1$ to $i-1$ to determine where to attach this node. 
 
 \begin{figure}[H]
    \centering
        \resizebox{0.4\textwidth}{!}{%
    \includegraphics{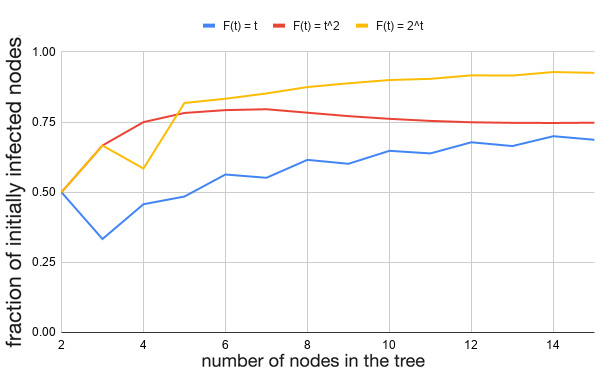}}
    \caption{Trials done on 10000 random trees of $n$ nodes, taking the average, and dividing it by $n$ for the fraction of node needed to be initially infected for the model to percolate.}
    \label{comp1}
\end{figure}

 In the above picture, the size of  the smallest minimal percolating set  can be obtained by multiplying the size of the minimal percolating set by the corresponding value of $n$. In particular, one can see how the exponential function requires an increasingly larger minimal percolating set in comparison with polynomial percolating functions. 
 \newpage

 \noindent{\bf Comparison with \cite{percset}.} 
To compare with the work of   \cite{percset}, we shall run the algorithm with $F(t)=2$ (leading to 2-bootstrap percolation as considered in  \cite{percset}) as well as linear-timed function on the following graph:

 \begin{figure}[H]
    \centering
        \resizebox{0.2\textwidth}{!}{%
    \includegraphics{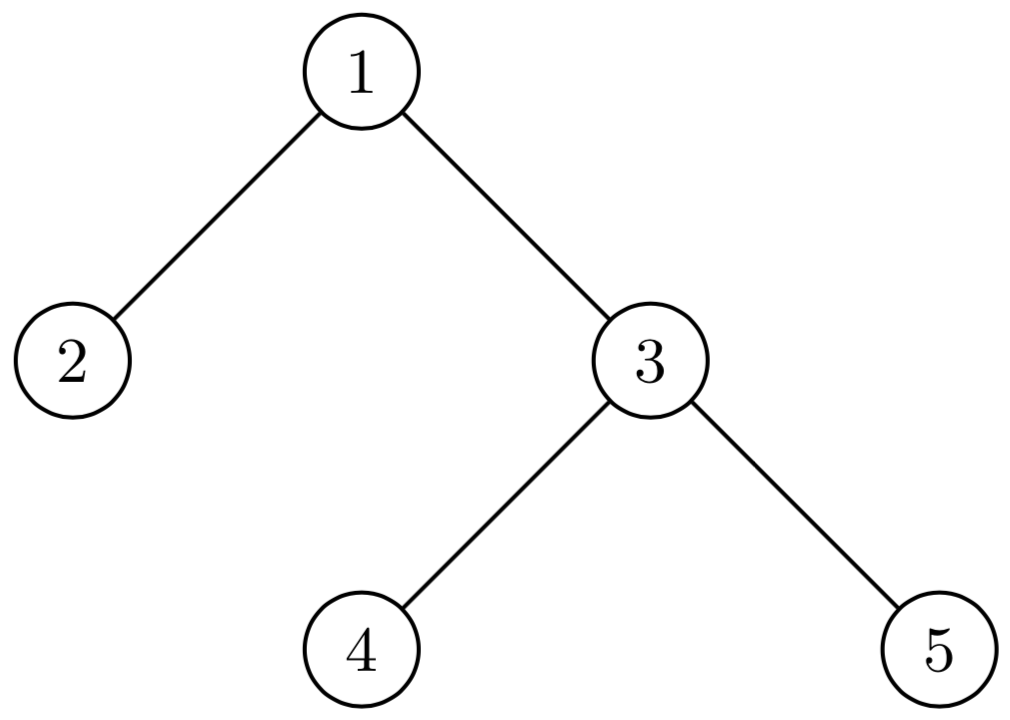}}
    \caption{Degree 2 tree with 5 nodes.}
    \label{comp1}
\end{figure}

With our algorithm, we see that nodes $2$, $3$ and $5$ are isolated respectively, and when we add them to the initial set, all nodes become infected. Thus the smallest minimal percolating set with our algorithm has size $3$.

Riedl provided an algorithm for the smallest minimal percolating sets in trees for $r$-bootstrap percolation in \cite{percset} that runs in linear time. We shall describe his algorithm generally to clarify the comparisons we will make. Riedl defined a trailing star or trailing pseudo-star as a subtree with each vertex being of distance at most $1$ or $2$ away, respectively, from a certain center vertex that is connected to the rest of the tree by only one edge. Then, the first step of Riedl's algorithm is a reduction procedure that ensures every non-leaf has degree at least $r$: intuitively, one repeatedly finds a vertex with degree less than $r$, include it to the minimal percolating set, remove it and all the edges attached to it, and for each of the connected components, add a new node with degree $1$ connected to the node that was a neighbor of the node we removed. 
Then, the algorithm identifies a trailing star or pseudo-star, whose center shall be denoted by $v$ and its set of leaves by $L$. Letting the original tree be $T$, if the number of leafs on $v$ is less than $r$, then set $T'=T \setminus (v \cup L)$; otherwise, set $T'=T\setminus L$. Recursively set $A'$ as the smallest minimal percolating set of $T'$ under $r$-bootstrap percolation. Then, the smallest minimal percolating set for $T$ is $A' \cup L$ if $|L|<r$ and $A' \cup L \setminus v$ otherwise.
Using Riedl's algorithm, we first note that there is a trailing star centered at $3$ with $2$ leaves. Removing the leaf, there is a trailing star at $1$ with $1$ leaf. Removing $1$ and $2$, we have one node left, which is in our $A'$. Adding the leaves back and removing $3$, we have an $A_0$ of $2,3$ and $5$, a smallest minimal percolating set. Thus the smallest minimal percolating set with Riedl's algorithm also has size $3$, as expected.

We shall now compare our algorithm to that of Riedl. A key step in Riedl's algorithm, which is including the leaves of stars and pseudo-stars in the final minimal percolating set, assumes that these leaves cannot be infected as it is assumed that $r > 1$. However, in our algorithm, we consider functions that may have the value of $1$ somewhere in the function, thus we cannot make that assumption. Further, in $r$-bootstrap percolation, time of infection of each vertex does not need to be taken into account when calculating the conditions for a node to be infected as that $r$ is constant, whereas in the time-dependent case, it is necessary: suppose a node has $n$ neighbors, and there is only one $t$ such that $F(t) \leq n$, so all neighbors must be infected by time $n$ in order for $n$ to become infected.\\

 \noindent{\bf Concluding remarks.}  The problem our algorithm solves is a generalization of Riedl's, for that it finds one smallest minimal percolating set for functions including constant ones. It has higher computational complexity for that it is not guaranteed for an unisolated node to be infected once one other neighbor of it is infected without accounting for time limits.
Finally, we should mention that the work presented in previous sections could be generalized in several directions and, in particular, we hope to develop a similar algorithm for largest minimal percolating set; and    study   the size of largest and smallest minimal percolating sets in  lattices.
 
  ~\\  
\noindent{\bf Acknowledgements.} The authors are thankful to MIT PRIMES-USA for the opportunity to conduct this research together, and in particular
Tanya Khovanova for her continued support, to Eric Riedl and Yongyi Chen for comments on a draft of the paper, and  to Rinni Bhansali and Fidel I. Schaposnik for useful advice regarding our code.   The work of Laura Schaposnik is partially supported through  the NSF grants
DMS-1509693 and CAREER DMS 1749013, and she is thankful to the Simons Center for Geometry and Physics for the hospitality during part of the preparation of the manuscript. This material is also based upon work supported by the National Science 
Foundation under Grant No. DMS- 1440140 while Laura Schaposnik was in residence at the Mathematical Sciences 
Research Institute in Berkeley, California, during the Fall 2019 semester.

\bibliography{Schaposnik_Percolation}{}

\newcommand{\etalchar}[1]{$^{#1}$}
\begin{thebibliography}{AMSP{\etalchar{+}}18}

\bibitem[AL03]{applications}
J.~Adler and U.~Lev, ``Bootstrap percolation: visualizations and
  applications,'' {\em Brazilian Journal of Physics} {\bfseries 33} no.~3,
  (2003) 641--644.

\bibitem[AMSP{\etalchar{+}}18]{ahmad2018analyzing}
N.~M. Ahmad, C.~Monta{\~n}ola-Sales, C.~Prats, M.~Musa, D.~L{\'o}pez, and
  J.~Casanovas-Garcia, ``analyzing policymaking for tuberculosis control in
  nigeria,'' {\em Complexity} {\bfseries 2018} (2018) .

\bibitem[BGH{\etalchar{+}}14]{MR3164766}
B.~Bollob\'{a}s, K.~Gunderson, C.~Holmgren, S.~Janson, and M.~Przykucki,
  ``Bootstrap percolation on {G}alton-{W}atson trees,'' {\em Electron. J.
  Probab.} {\bfseries 19} (2014) no. 13, 27.

\bibitem[Bra17]{epidemiology}
F.~Brauer, ``Mathematical epidemiology: Past, present, and future,'' {\em
  Infectious Disease Modelling} {\bfseries 2} no.~2, (2017) 113--127.

\bibitem[BS19]{gossip}
R.~Bhansali and L.~P. Schaposnik, ``A trust model for spreading gossip in
  social networks,'' {\em arXiv preprint arXiv:1905.11204} (2019) .

\bibitem[CR18]{chowell2018spatial}
G.~Chowell and R.~Rothenberg, ``Spatial infectious disease epidemiology: on the
  cusp,'' 2018.

\bibitem[DAPB{\etalchar{+}}15]{challenges2}
D.~De~Angelis, A.~M. Presanis, P.~J. Birrell, G.~S. Tomba, and T.~House, ``Four
  key challenges in infectious disease modelling using data from multiple
  sources,'' {\em Epidemics} {\bfseries 10} (2015) 83--87.

\bibitem[FBB{\etalchar{+}}15]{challenges}
S.~Funk, S.~Bansal, C.~T. Bauch, K.~T. Eames, W.~J. Edmunds, A.~P. Galvani, and
  P.~Klepac, ``Nine challenges in incorporating the dynamics of behaviour in
  infectious diseases models,'' {\em Epidemics} {\bfseries 10} (2015) 21--25.

\bibitem[GGM15]{greenhalgh2015disease}
S.~Greenhalgh, A.~P. Galvani, and J.~Medlock, ``Disease elimination and
  re-emergence in differential-equation models,'' {\em Journal of theoretical
  biology} {\bfseries 387} (2015) 174--180.

\bibitem[HMNK19]{hunter2019correction}
E.~Hunter, B.~Mac~Namee, and J.~Kelleher, ``Correction: An open-data-driven
  agent-based model to simulate infectious disease outbreaks,'' {\em PloS one}
  {\bfseries 14} no.~1, (2019) e0211245.

\bibitem[IAC{\etalchar{+}}15]{imai2015time}
C.~Imai, B.~Armstrong, Z.~Chalabi, P.~Mangtani, and M.~Hashizume, ``Time series
  regression model for infectious disease and weather,'' {\em Environmental
  research} {\bfseries 142} (2015) 319--327.

\bibitem[KL82]{density}
P.~M. Kogut and P.~Leath, ``High-density site percolation on real lattices,''
  {\em Journal of Physics C: Solid State Physics} {\bfseries 15} no.~20, (1982)
  4225.

\bibitem[KM27]{kermack1927contribution}
W.~O. Kermack and A.~G. McKendrick, ``A contribution to the mathematical theory
  of epidemics,'' {\em Proceedings of the royal society of london. Series A,
  Containing papers of a mathematical and physical character} {\bfseries 115}
  no.~772, (1927) 700--721.

\bibitem[Mor09]{Morris}
R.~Morris, ``Minimal percolating sets in bootstrap percolation,'' {\em the
  electronic journal of combinatorics} {\bfseries 16} no.~1, (2009) R2.

\bibitem[PTM17]{pipatsart2017stochastic}
N.~Pipatsart, W.~Triampo, and C.~Modchang, ``Stochastic models of emerging
  infectious disease transmission on adaptive random networks,'' {\em
  Computational and mathematical methods in medicine} {\bfseries 2017} (2017) .

\bibitem[Rie12]{percset}
E.~Riedl, ``Largest and smallest minimal percolating sets in trees,'' {\em the
  electronic journal of combinatorics} {\bfseries 19} no.~1, (2012) 64.

\bibitem[Ros11]{ross}
R.~Ross, ``The prevention of malaria (with addendum),''(1911) .

\bibitem[SN13]{dengue}
S.~Side and M.~S.~M. Noorani, ``A sir model for spread of dengue fever disease
  (simulation for south sulawesi, indonesia and selangor, malaysia),'' {\em
  World Journal of Modelling and Simulation} {\bfseries 9} no.~2, (2013)
  96--105.

\bibitem[SR13]{modeling}
C.~I. Siettos and L.~Russo, ``Mathematical modeling of infectious disease
  dynamics,'' {\em Virulence} {\bfseries 4} no.~4, (2013) 295--306.

\bibitem[VRMP10]{viral}
N.~K. Vaidya, R.~M. Ribeiro, C.~J. Miller, and A.~S. Perelson, ``Viral dynamics
  during primary simian immunodeficiency virus infection: effect of
  time-dependent virus infectivity,'' {\em Journal of virology} {\bfseries 84}
  no.~9, (2010) 4302--4310.

\bibitem[WDRS07]{automata}
S.~H. White, A.~M. Del~Rey, and G.~R. S{\'a}nchez, ``Modeling epidemics using
  cellular automata,'' {\em Applied Mathematics and Computation} {\bfseries
  186} no.~1, (2007) 193--202.

\bibitem[WEV{\etalchar{+}}09]{watkins2009disease}
R.~E. Watkins, S.~Eagleson, B.~Veenendaal, G.~Wright, and A.~J. Plant,
  ``Disease surveillance using a hidden markov model,'' {\em BMC medical
  informatics and decision making} {\bfseries 9} no.~1, (2009) 39.

\end{thebibliography}
\bibliographystyle{fredrickson}

%
%
%
%



\end{document}